\documentclass[12pt,oneside]{article}
\usepackage[inner=32mm,outer=31mm,tmargin=30mm,bmargin=40mm]{geometry}
\usepackage[utf8]{inputenc}
\usepackage[T1]{fontenc}
\usepackage{sectsty}
    \sectionfont{\large}
    \subsectionfont{\normalsize}
\usepackage{amssymb}
\usepackage{geometry}
\usepackage{amsmath}
\usepackage{amsthm}
\usepackage{graphicx}
\usepackage{sidecap}
\usepackage{color}
\usepackage[normalem]{ulem}
\usepackage[english]{babel}
\usepackage[onehalfspacing]{setspace}
\usepackage{verbatim}
\usepackage[hidelinks]{hyperref}

\newcommand{\IR}{\ensuremath{\mathbb{R}}}
\newcommand{\IN}{\ensuremath{\mathbb{N}}}

\newcommand{\set}[1]{\left\{#1\right\}}
\newcommand{\abs}[1]{\left|#1\right|}
\newcommand{\brackets}[1]{\left(#1\right)}

\newtheorem*{thm}{Theorem}

\theoremstyle{plain}
\newtheorem{lemma}{Lemma}

\theoremstyle{definition}

\title{On the Dispersion of Sparse Grids} 

\author{
David Krieg\\ 
}

\date{}

\begin{document}

\maketitle

\begin{abstract}
\noindent
 For any natural number $d$ and positive number $\varepsilon$,
 we present a point set in the $d$-dimensional unit cube $[0,1]^d$
 that intersects every axis-aligned box of volume greater than $\varepsilon$.
 These point sets are very easy to handle and
 in a vast range for $\varepsilon$ and $d$,
 we do not know any smaller set with this property.
\end{abstract}

\section{The Result}

The dispersion of a point set $P$ in $[0,1]^d$
is the volume of the largest axis-aligned box in $[0,1]^d$
which does not intersect~$P$. %It is denoted by $\disp(P)$. %The concept of the largest empty box amidst a point set ...
Point sets with small dispersion already proved to be useful
for the uniform recovery of rank one tensors~\cite{nr}
and for the discretization of the uniform norm of
trigonometric polynomials~\cite{t}.
Recently, great progress has been made in the question
for the minimal size for which there \emph{exists} a point set
whose dispersion is at most $\varepsilon>0$, see
Dumitrescu and Jiang~\cite{dj}, Aistleitner, Hinrichs and Rudolf~\cite{ahr},
Rudolf~\cite{r} and Sosnovec~\cite{s}.
In this note, we want to \emph{provide} such point sets.
They should be both small and easy to handle.
To that end, we define the point sets
\begin{equation*}
 P(k,d)=\bigcup_{\abs{\boldsymbol j}=k}
 M_{j_1} \times\dots\times M_{j_d}
\end{equation*}
of order $k\in\IN_0$ and dimension $d\in\IN$, generated by the one-dimensional sets
\begin{equation*}
 M_j=\set{\frac{1}{2^{j+1}},\frac{3}{2^{j+1}},\hdots,\frac{2^{j+1}-1}{2^{j+1}}}
  \quad
 \text{for}
 \quad
 j\in\IN_0.
\end{equation*}
You can find a picture of the set of order $3$ in dimension $2$ in Figure~\ref{Point Set}.
These point sets are particular instances of a sparse grid
as widely used for high-dimensional numerical integration
and approximation. We refer to
Novak and Woźniakowski~\cite{track2} and the references therein.
Here, we will prove the following result.

\begin{thm}
 Let $k(\varepsilon)=\left\lceil \log_2\brackets{\varepsilon^{-1}} \right\rceil -1$
 for any $\varepsilon\in (0,1)$ and let $d\geq 2$.
 Then the dispersion of $P(k(\varepsilon),d)$ is at most $\varepsilon$
 and
 \begin{equation*}
  \abs{P(k(\varepsilon),d)} =
  2^{k(\varepsilon)}\, \binom{d+k(\varepsilon)-1}{d-1}.
 \end{equation*}
\end{thm}

The formula for the size of $P(k(\varepsilon),d)$ may be simplified.
On the one hand, we have
\begin{equation*}
 \abs{P(k(\varepsilon),d)}
 \leq \varepsilon^{-1}\left\lceil \log_2\brackets{\varepsilon^{-1}} \right\rceil^{d-1},
\end{equation*}
which shows that the size grows almost linearly in $1/\varepsilon$ for a fixed dimension $d$.
On the other hand,
\begin{equation*}
 \abs{P(k(\varepsilon),d)}
 \leq (2d)^{k(\varepsilon)},
\end{equation*}
which shows that the size only grows polynomially in $d$ for a fixed error tolerance~$\varepsilon$.
Although very simple, $P(k(\varepsilon),d)$ is the smallest explicitly known
point set in $[0,1]^d$ with dispersion at most $\varepsilon$
for many instances of $\varepsilon$ and $d$, see Section~\ref{Comparison Section}.

%$\leq \min\set{\varepsilon^{-1}\left\lceil \log_2\brackets{\varepsilon^{-1}} \right\rceil^{d-1},
%  (2d)^{k(\varepsilon)}}$

\section{The Proof}

Let us introduce the notation used in this section.
We write $[d]=\set{1,\hdots,d}$ for each $d\in\IN$
and $\abs{\boldsymbol j}=j_1+\hdots+j_d$ for $\boldsymbol j\in\IN_0^d$.
The vector $\boldsymbol e_\ell\in\IR^d$ has the entry 1 in the $\ell$th coordinate
and 0 in all other coordinates.
A set $B\subset \IR^d$ is called a box,
if it is the Cartesian product of $d$ open intervals.
Its volume $\abs{B}$ is the product of their lengths.
If $M$ is a finite set, $\abs{M}$ denotes the number of its elements.
We start with computing the number of elements in $P(k,d)$ for $k\in\IN_0$
and $d\in\IN$.

\begin{lemma}
 \begin{equation*}
  \abs{P(k,d)}=2^k\, \binom{d+k-1}{d-1}.
 \end{equation*}
\end{lemma}

\begin{proof}
 Note that $\abs{M_j}=2^j$ for all $j\in\IN$ and also for $j=0$.
 The identity
 \begin{equation*}
  \abs{P(k,d)}=
  \sum_{\abs{\boldsymbol j}=k}
  \abs{M_{j_1}\times\hdots\times M_{j_d}} 
  = \sum_{\abs{\boldsymbol j}=k}
  2^{j_1+\hdots+j_d}
  = 2^k\, \abs{\set{\boldsymbol j\in\IN_0^d: \abs{\boldsymbol j}=k}}
 \end{equation*}
 yields the statement of the lemma.
\end{proof}

It follows from \cite{t} that the dispersion of $P(k,d)$
decays like $2^{-k}$, if $d$ is fixed and $k$ tends to infinity.
In fact, it can be computed precisely.
In dimension $d=1$, it is easy to see that the dispersion of $P(k,d)$
equals $2^{-k}$ for $k\geq 1$ and $1/2$ for $k=0$.
In dimension $d\geq 2$, we obtain the following.

\begin{lemma}
For any $k\in\IN_0$ and $d\geq 2$, the dispersion of $P(k,d)$ is $2^{-(k+1)}$.
\end{lemma}

\begin{SCfigure}
  \caption{The Point Set $P(3,2)$.
  This picture shows the set $P(k,d)$ of order $3$
  in dimension $2$.
  The largest empty box has the volume $1/16$,
  the size of $16$ of the little squares.
  If any of the $32$ points is removed,
  an empty box of volume $1/8$ emerges.}
  \includegraphics[width=0.65\textwidth]{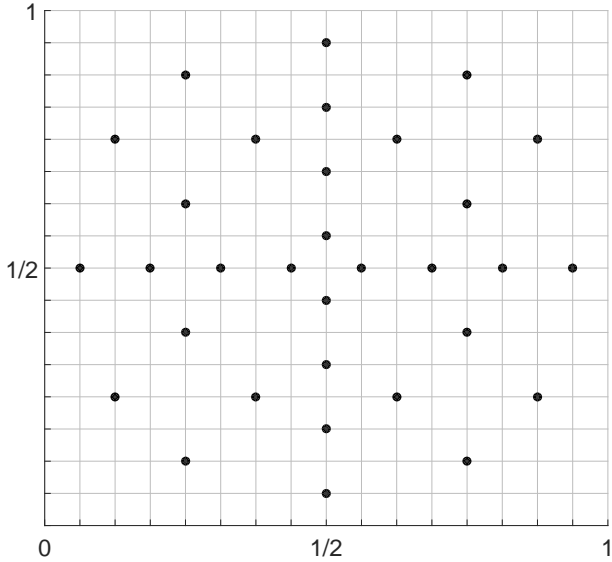}
  \label{Point Set}
\end{SCfigure}

\begin{comment}
\begin{figure}
\begin{minipage}{.6\linewidth}
\includegraphics[width=\linewidth]{SparseGrid.eps}
\end{minipage}
\hspace*{.05\linewidth}
\begin{minipage}{.28\linewidth}
\caption{The Point Set $P(3,2)$.}
This picture shows the set $P(k,d)$ of level $3$
in dimension $2$.
The largest empty box has the volume $1/16$,
the size of $16$ of the little squares.
If any of the $32$ points is removed,
an empty box of volume $1/8$ emerges.
\end{minipage}
\label{Point Set}
\end{figure}
\end{comment}

\begin{proof}
We first observe that there are many boxes of volume $2^{-(k+1)}$
which do not intersect with $P(k,d)$. For instance, the box
\begin{equation*}
 (0,2^{-(k+1)})\times(0,1)\times\dots\times(0,1)
\end{equation*}
has these properties.
Now, let $B=I_1\times\dots\times I_d$ be any box in $[0,1]^d$
that does not intersect with $P(k,d)$. The set
\begin{equation*}
 P= \bigcup_{m\in\IN} P(m,d)
 = \set{\frac{\alpha}{2^\beta} \mid 
 \beta \in\IN \text{ and } \alpha\in \left[2^\beta-1\right]}^d
\end{equation*}
is dense in $[0,1]^d$.
Without loss of generality, we assume that the interior of $B$ is nonempty.
Therefore, $B$ intersects with $P$ and hence with $P(m,d)$ for some~$m\in\IN$.
Let $m$ be minimal with this property.
Since $B$ does not intersect with $P(k,d)$, we either have $m>k$ or $m<k$.
Let $\boldsymbol x\in P(m,d)\cap B$.
This means that there is some $\boldsymbol j\in\IN_0^d$ with
$\abs{\boldsymbol j}=m$ and
\begin{equation*}
 x_\ell \in M_{j_\ell} \cap I_\ell
\end{equation*}
for all $\ell\in[d]$.
We observe that the numbers
$x_\ell \pm \frac{1}{2^{j_\ell+1}}$ are either
contained in $\set{0,1}$ or in $M_j$ for some $j< j_\ell$.
Hence, they are not contained in $I_\ell$,
because $I_\ell$ is a subset of $(0,1)$ and $m$ is minimal.
We obtain that 
\begin{equation*}
 I_\ell \subseteq
 \brackets{x_\ell - \frac{1}{2^{j_\ell+1}},x_\ell + \frac{1}{2^{j_\ell+1}}},
\end{equation*}
and hence
\begin{equation*}
 \abs{B} \leq \prod_{\ell\in[d]} 2^{-j_\ell} = 2^{-m}.
\end{equation*}
In the case $m>k$, this yields the statement.
In the case $m<k$, we observe that the
numbers
$x_\ell \pm \frac{1}{2^{k-m+j_\ell+1}}$
cannot be contained in $I_\ell$ for any $\ell\in[d]$,
since otherwise the points
\begin{equation*}
\boldsymbol x \pm \frac{\boldsymbol e_\ell}{2^{k-m+j_\ell+1}}
\end{equation*}
would be both in $B$ and in $P(k,d)$.
This means that
\begin{equation*}
 I_\ell \subseteq
 \brackets{x_\ell - \frac{1}{2^{k-m+j_\ell+1}},x_\ell + \frac{1}{2^{k-m+j_\ell+1}}}.
\end{equation*}
We obtain
\begin{equation*}
 \abs{B} \leq 
 \prod_{\ell\in[d]} 2^{m-k-j_\ell} 
 = 2^{dm-dk-m}
 \leq 2^{-(k+1)},
\end{equation*}
where we used that $d\geq 2$, and the statement is proven.
\end{proof}

For any $\varepsilon\in (0,1)$, the smallest number $k\in\IN_0$
that satisfies $2^{-(k+1)}\leq\varepsilon$ is obviously given by
\begin{equation*}
 k(\varepsilon)=\left\lceil \log_2\brackets{\varepsilon^{-1}} \right\rceil -1.
\end{equation*}
This yields the statement of our theorem.

\section{A Comparison with Known Results}
\label{Comparison Section}

Let $d\geq 2$ be an integer and $\varepsilon\leq 1/4$ be a positive number.
Let us call a point set in $[0,1]^d$ admissible,
if its dispersion is at most $\varepsilon$.
In 2015, Aistleitner, Hinrichs and Rudolf~\cite{ahr} proved
that any admissible point set contains at least
\begin{equation}
\label{lower bound}
 (4\varepsilon)^{-1} (1-4\varepsilon)\log_2 d
\end{equation} 
points. At that time, the smallest known admissible point set
was a Halton-Hammersley set of size
\begin{equation}
\label{hammersley}
 \left\lceil 2^{d-1} \pi_d\, \varepsilon^{-1}\right\rceil,
\end{equation}
where $\pi_d$ is the product of the first $(d-1)$ primes.
This was proven by Rote and Tichy~\cite{rt}, see also 
Dumitrescu an Jiang~\cite{dj} for more details.
The size of this set grows as slowly as possible as $\varepsilon$
tends to zero, if $d$ is fixed.
However, it grows super-exponentially with $d$.
Gerhard Larcher realized that also certain $(t,m,s)$-nets
that contain only
\begin{equation}
\label{nets}
 \left\lceil 2^{7d+1} \varepsilon^{-1}\right\rceil,
\end{equation}
points are admissible. The proof is included in~\cite{ahr}.
This number is smaller than (\ref{hammersley}) for $d\geq 54$.
However, its exponential growth with respect to $d$ for fixed $\varepsilon$
is still far away from the logarithmic growth of the lower bound (\ref{lower bound}).
In the beginning of 2017, Rudolf~\cite{r} significantly narrowed
this gap. In line with % With reference to / By virtue of
results of Blumer, Ehrenfeucht, Haussler and Warmuth~\cite{behw},
he obtained the existence of an admissible point set of size
\begin{equation}
\label{Daniel}
 \left\lfloor 8d\, \varepsilon^{-1} \log\brackets{33\varepsilon^{-1}} \right\rfloor.
\end{equation}
Very recently, the remaining gap was completely closed by Sosnovec~\cite{s},
who proved the existence of an admissible point set of size
\begin{equation}
 \label{Jacub}
 \left\lfloor q^{q^2+2}(1+4\log q)\cdot \log d \right\rfloor,
\end{equation}
where $q=\lceil 1/\varepsilon\rceil$.
\begin{comment}
 On the other hand, the $\varepsilon$-dependency of this number is very bad,
such that (\ref{Daniel}) is still the smallest known upper bound
for the minimal size
$N_*(\varepsilon,d)$ %=\min\set{\abs{P} \mid P \text{ admissible}}
of admissible point sets for a vast majority 
of parameters $\varepsilon~\geq~10^{-12}$ and $d~\leq~10^{12}$.
\end{comment}
Just like the lower bound (\ref{lower bound}), this number
only grows logarithmically with $d$.
On the other hand, it now depends super-exponentially on $1/\varepsilon$.
This can most likely be improved,
but up to now, (\ref{Daniel}) is still the smallest known upper bound
for the minimal size
$N_*(\varepsilon,d)$ %=\min\set{\abs{P} \mid P \text{ admissible}}
of admissible point sets for a vast majority 
of parameters $\varepsilon~\geq~10^{-12}$ and $d~\leq~10^{12}$.
However, the upper bounds (\ref{Daniel}) and (\ref{Jacub}) are both
nonconstructive in the sense
that no admissible point set was provided.
% only the upper bounds (\ref{hammersley}) and (\ref{nets}) are constructive.
Here, we proved the existence of an admissible point set of size
\begin{equation}
 \label{new}
 2^{k(\varepsilon)}\, \binom{d+k(\varepsilon)-1}{d-1}
\end{equation}
where $k(\varepsilon)=\lceil \log_2\brackets{\varepsilon^{-1}}\rceil -1$,
which does not grow exponentially with $d$ either,
but we also provided such a set, namely $P(k(\varepsilon),d)$.
Already for modest dimensions, this number is much smaller than
the cardinalities (\ref{hammersley}) and (\ref{nets}).
As an example, we consider the parameters $d\in\set{2,\hdots,100}$ and
$\varepsilon\in\set{1/4,1/5,\dots,1/100}$ in Figure~\ref{comparison figure}.
The dark gray area represents the parameters for which the author
does not know any smaller admissible set than $P(k(\varepsilon),d)$,
although Rudolf~\cite{r} proved their existence.
There are also some parameters, where $P(k(\varepsilon),d)$ improves
on all the mentioned upper bounds for $N_*(\varepsilon,d)$,
constructive or nonconstructive. %\widowpenalty1000
\mbox{These parameters are indicated} by the black area.
In the light gray area, the Halton-Hammersley point set is smaller.

\begin{figure}
 \caption{Minimality Properties of $P(k(\varepsilon),d)$.}
 \includegraphics[width=\textwidth]{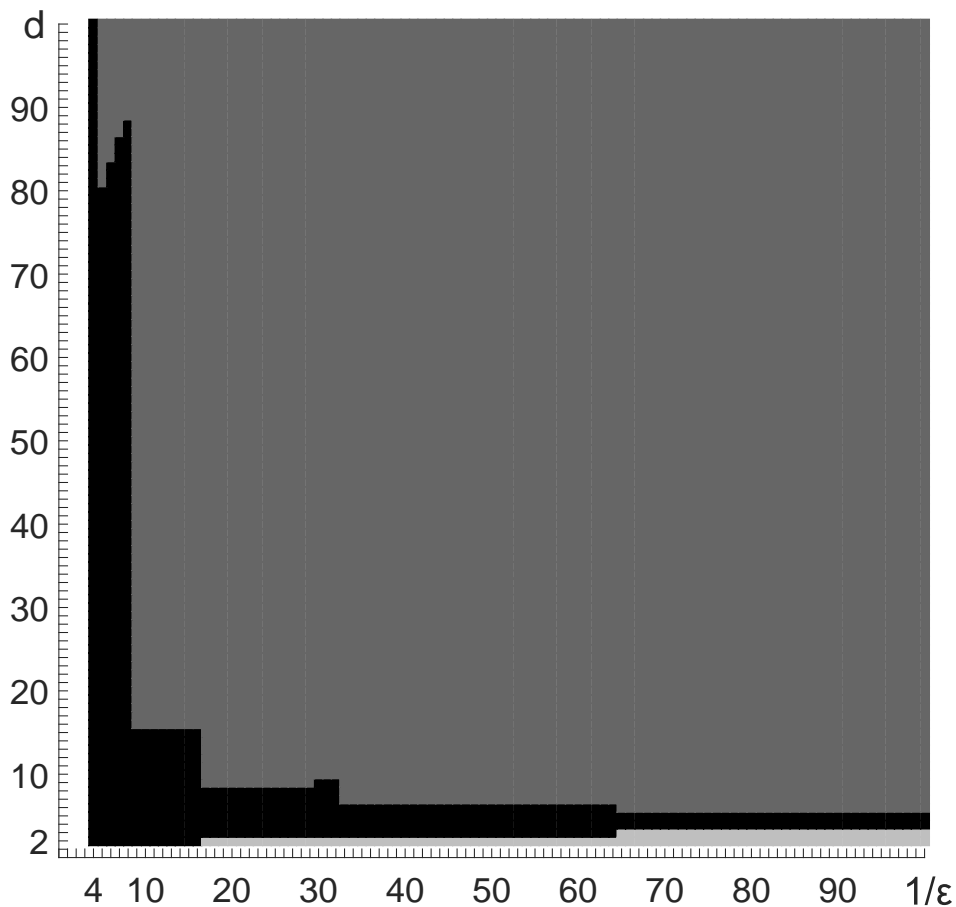}
 \label{comparison figure}
\end{figure}

\begin{comment}
\begin{SCfigure}
  \caption{Minimality Properties of $P(k(\varepsilon),d)$.}
  \includegraphics[width=0.8\textwidth]{comparison_in_gray_100.eps}
  \label{comparison figure}
\end{SCfigure}
alt:
\caption{The Minimality of $P(k(\varepsilon),d)$.
 For $(\varepsilon,d)$ from the black area,
 it is unclear, whether there is a smaller admissible point set.
 For $(\varepsilon,d)$ from the dark gray area,
 no such set has yet been provided.}
 
\begin{figure}
\begin{minipage}{.7\linewidth}
 \includegraphics[width=\linewidth]{comparison_in_grey.eps}
\end{minipage}
\begin{minipage}{.19\linewidth}
 \caption{The Minimality of $P(k(\varepsilon),d)$.}
 For $(\varepsilon,d)$ from the black area,
 it is unclear, whether there is a smaller admissible point set.
 For $(\varepsilon,d)$ from the dark gray area,
 no such set has yet been provided.
\end{minipage}
\label{comparison figure}
\end{figure}
\end{comment}

\raggedright{

}


\begin{thebibliography}{XX}
 
 \bibitem{ahr} C.\,Aistleitner, A.\,Hinrichs, D.\,Rudolf:
 \textit{On the Size of the Largest Empty Box amidst a Point Set}.
 Discrete Applied Mathematics~\textbf{230}, 146--150, 2017.
 
 \bibitem{behw} A.\,Blumer, A.\,Ehrenfeucht, D.\,Haussler, M.\,Warmuth: 
 \textit{Learnability  and  the  Vapnik-Chervonenkis dimension.}
 J.\,Assoc.\,Comput.\,Mach.\,\textit{36}(4), 929--965, 1989.
 
 \bibitem{dj} A.\,Dumitrescu, M.\,Jiang:
 \textit{On the Largest Empty Axis-Parallel Box amidst $n$ Points}.
 Algorithmica~\textbf{66}(2), 225--248, 2013.
 
 \bibitem{nr} E.\,Novak, D.\,Rudolf:
 \textit{Tractability of the Approximation of High-Dimensional Rank One Tensors}.
 Constructive Approximation~\textbf{43}, 1--13, 2016.
 
 \bibitem{track2} E.\,Novak, H.\,Woźniakowski:
 \textit{Tractability of Multivariate Problems. Volume II: Standard Information for Functionals}.
 EMS, Zürich, 2010.
 
 \bibitem{rt} G.\,Rote, R.\,F.\,Tichy.
 \textit{Quasi-Monte Carlo Methods and the Dispersion of Point Sequences.}
 Math.\,Comput.\,Modelling, \textbf{23}(8-9), 9--23, 1996.
 
 \bibitem{r} D.\,Rudolf: \textit{An upper Bound on the Minimal Dispersion via Delta Covers}.
 Festschrift for Ian Sloan's 80th Birthday, to appear.
 Available at ArXiv e-prints, arXiv:1701.06430 [cs.CG].
 
 \bibitem{s} J.\,Sosnovec: \textit{A Note on the Minimal Dispersion of Point Sets
 in the Unit Cube}.
 ArXiv e-prints, 2017, arXiv:1707.08794 [cs.CG].
 
 \bibitem{t} V.\,N.\,Temlyakov: \textit{Universal Discretization}.
 ArXiv e-prints, 2017, arXiv:1708.08544 [math.NA].

\end{thebibliography}
\end{document}